\documentclass[preprint]{imsart}

\setattribute{journal}{name}{}
\usepackage[utf8]{inputenc}
\usepackage[colorlinks,citecolor=blue,urlcolor=blue,linkcolor=magenta]{hyperref}

\usepackage[abs]{overpic}

\usepackage{comment}

\usepackage[authoryear]{natbib}
\usepackage{times}

\usepackage{multirow}

\usepackage{amssymb,amsmath, amsthm}
\usepackage{graphicx}

\usepackage{framed} 

\usepackage[english]{babel}
\usepackage{booktabs}
\usepackage{color}
\usepackage{epsfig}

\usepackage{natbib}
\newcommand{\si}{\sigma}

\renewcommand{\th}{\theta}

\newcommand{\ga}{\gamma}

\newcommand{\eps}{\varepsilon}

\renewcommand{\phi}{\varphi}

\newcommand{\scr}[1]{{\mathcal #1}}

\newcommand{\EE}{\mathbb{E}}

\newcommand{\PP}{\mathbb{P}}

\newcommand{\RR}{\mathbb{R}}

\renewcommand{\tilde}{\widetilde}

\providecommand{\trace}{{\operatorname{tr}}}

\newcommand{\Ell}{{\mathcal L}}

\newcommand{\rara}[1]{\renewcommand{\arraystretch}{#1}}

\newcommand{\en}{\mathrm{end}}

\mathchardef\given="626A

\newcommand{\bem}{\begin{bmatrix}}
\newcommand{\enm}{\end{bmatrix}}

\newtheorem{thm}{Theorem}[section]
\newtheorem{lemma}[thm]{Lemma}
\newtheorem{corr}[thm]{Corollary}
\theoremstyle{definition}
\newtheorem{algorithm}{Algorithm}

\newtheorem{rem}[thm]{Remark}
\newtheorem{ex}[thm]{Example}
\newtheorem{defn}[thm]{Definition}

\newtheorem{ass}[thm]{Assumption}

\newcommand{\dd}{{\,\mathrm d}}

\newcommand{\DD}{{\,\mathrm D}}
\newcommand{\T}{{\prime}}

\newcommand{\even}{\text{even}}
\newcommand{\odd}{\text{odd}}

\begin{document}

\begin{frontmatter}

\title{Bayesian estimation of incompletely  observed diffusions}

\runtitle{Bayesian estimation of diffusions}

\begin{aug}
\author{Frank van der Meulen\ead[label=e1]{f.h.vandermeulen@tudelft.nl} and
Moritz Schauer\ead[label=e2]{m.r.schauer@math.leidenuniv.nl}}

\runauthor{Van der Meulen and Schauer}

\affiliation{Delft University of Technology}

\address{Delft Institute of Applied Mathematics (DIAM) \\
Delft University of Technology\\
Mekelweg 4\\
2628 CD Delft\\
The Netherlands\\
\printead{e1}\\[1em]
Mathematical Institute\\
Leiden University\\
P.O. Box 9512\\
2300 RA Leiden\\
The Netherlands\\
\printead{e2}}

\end{aug}

\begin{abstract}.
We present a general framework for Bayesian estimation of incompletely observed multivariate diffusion processes. Observations are assumed to be discrete in time, noisy and incomplete. We assume the drift and diffusion coefficient depend on an unknown parameter. A data-augmentation algorithm for drawing from the posterior distribution is presented which is based on simulating  diffusion bridges conditional on a noisy incomplete observation at an intermediate time. The dynamics of such filtered bridges are derived and it is shown how these can be simulated using a generalised version of the guided proposals introduced in \cite{Schauer}.

\medskip

\noindent
 \emph{ {Keywords:} data augmentation; enlargement of filtration; guided proposal; filtered bridge; smoothing diffusion processes; innovation scheme;  Metropolis-Hastings; multidimensional diffusion bridge; partially observed diffusion}
 \end{abstract}

\begin{keyword}[class=MSC]
\kwd[Primary ]{62M05, 60J60}
\kwd[; secondary ]{ 62F15}
\kwd{65C05}
\end{keyword}
%

\end{frontmatter}

\numberwithin{equation}{section}

\section{Introduction}

We consider Bayesian estimation for incompletely, discretely observed multivariate diffusion processes. 
Suppose $X$ is a multidimensional diffusion with time dependent drift $b\colon\, \RR_+\times \RR^d \to \RR^d$ and time dependent dispersion coefficient $\si\colon\, \RR_+\times \RR^d \to \RR^{d\times d'}$ governed by the stochastic differential equation (SDE) 
\begin{equation}\label{eq:sde} \dd X_t = b(t,X_t) \dd t + \si(t,X_t) \dd W_t. \end{equation}
The process $W$ is a vector valued process in $\RR^{d'}$ consisting of independent Brownian motions. 
Denote observation times by $0=t_0 < t_1<\cdots < t_n$. 
Denote $X_i\equiv X_{t_i}$ and assume observations 
\[	V_i=L_i X_{i} +\eta_i,\qquad i=0,\ldots, n,  \] where $L_i$ is a  $m_i \times d$-matrix. The random variable $\eta_i$ is assumed to have a continuous density $q_i$, which may for example  be the $N_{m_i}(0,\Sigma_i)$-density. Further, we assume $\eta_0,\ldots, \eta_n$ is a sequence of independent random variables, independent of the diffusion process $X$. 
This setup includes full observations in case $L_i=I_d$ (the identity matrix of dimension $d\times d$). Further, if $m_k<d$ we have  observations  that are in a plane of dimension strictly smaller than $d$, with error superimposed. Suppose $b$ and $\si$ depend on an unknown finite dimensional parameter $\th$.  Based on the information set
\[	\scr{D}:=\{V_i,\, i=0,\ldots, n\} \] we wish to infer $\th$ within the Bayesian paradigm. 

From an applied point of view, there are many motivating examples that correspond to the outlined problem.  As a first example,  in chemical kinetics the evolution of concentrations of particles of different species is modeled by  stochastic differential equations. In case it is only possible to  measure the cumulative concentration of two species but not the single concentrations,  we have incomplete observations with $L=\begin{bmatrix} 1 & 1\end{bmatrix}$. A second example is given by stochastic volatility models used in finance, where the  volatility process is unobserved. If the price of an asset is the first component of the model and the latent volatility the second component, then we have incomplete observations with $L=\begin{bmatrix} 1 & 0\end{bmatrix}$. Note that in our setup the way in which the observations are incomplete need not be the same at all observation times (that is, $L_i$ may differ from $L_j$ for $i\neq j$). Hence, missing data fit naturally within our framework. 

\subsection{Related work}
Even in case of full discrete time observations the described problem is hard as no closed form expression for the likelihood can be written down, aside from some very specific easy cases.  To work around this problem, data-augmentation has been proposed where the latent data are the missing diffusion bridges that connect the discrete time observations. See for instance \cite{RobertsStramer}, \cite{ChibPittShephard}, \cite{BeskosPapaspiliopoulosRobertsFearnhead}, \cite{Voss}, \cite{MR2422763}, \cite{GolightlyWilkinsonChapter},  \cite{Fuchs}, \cite{PapaRobertsStramer}, and \cite{Schauer2}.   The resulting algorithm has been shown to be successful provided one is able to draw diffusion bridges between two adjacent discrete time observations efficiently. A major simplification that the fully observed case brings  is that diffusion bridges can be simulated independently. The latter property is lost in case of incomplete observations: the latent process between times $t_{i-1}$ and $t_i$ depends  on {\it all} observations  $V_0, V_{1},\ldots, V_n$.  This dependence may seem to imply that it is infeasible to draw such diffusion bridges. Indeed this is hard, but is in fact not necessary as  we can draw $(X_t,\, t\in [0,T])$ in blocks.  This idea has appeared in several papers. Both \cite{MR2422763} and \cite{Fuchs}  consider the case where $L_i=I_d$  with possibly several rows removed (which corresponds to not observing corresponding components of the diffusion).  
For $i<j$   set  $X_{(i:j)}=\{X_t,\, t\in (t_i, t_j\}$. \cite{MR2422763} discretise the SDE and construct an algorithm according to the steps:
\begin{enumerate}
\item Initialise $X_{(0:n)}$ and $\th$. 
\item For $i=0,\ldots, n-2$, sample filtered diffusion bridges $X_{(i: i+2)}$, conditional on $X_{i}$, $V_{i+1}$, $X_{i+2}$  and $\th$. Sample $X_{(0:1)}$ conditional on $V_0$, $X_1$ and $\th$. Sample $X_{(n-1:n)}$ conditional on $X_{n-1}$, $V_n$ and $\th$. 
\item Sample $\th$ conditional on  $X_{(0:n)}$.
\end{enumerate}
In fact, the second step is carried out slightly differently using the ``innovation scheme'', as we will discuss shortly (moreover, updating the first and last segment requires special care). 
\cite{Fuchs} (section 7.2) proposes a similar algorithm using  some variations on carrying out the second step. 
In both references, bridges are proposed based on the Euler discretisation of the SDE for $X$ with $b\equiv 0$ and accepted using the Metropolis-Hastings rule. In case of either strong nonlinearities in the drift or low sampling frequency this can lead to very low acceptance probabilities. 

A diffusion bridge is an infinite-dimensional random variable. The approach taken in \cite{MR2422763} and \cite{Fuchs} is to approximate this stochastic process by a finite-dimensional vector and next carry out simulation. \cite{PapaspiliopoulosRoberts} call this the projection-simulation strategy and advocate the simulation-projection strategy where an appropriate Monte-Carlo scheme is designed that operates on the infinitely-dimensional space of diffusion bridges. For practical purposes it needs to be discretised but the  discretisation error can be eliminated by letting the mesh-width tend to zero. This implies that the algorithm is valid when taking this limit. We refer to \cite{PapaspiliopoulosRoberts} to a  discussion on additional advantages of the simulation-projection strategy, which we will employ in this paper. 

Within the simulation-projection setup a particular version of the problem in this article has been treated  in   the unpublished Ph.D.\ thesis \cite{Jensen} (chapter 6). Here, it is assumed that certain components of the diffusion are unobserved, whereas the remaining components are observed discretely without error. A major limitation of this work is that it is essential that the diffusion can be transformed to unit diffusion coefficient. 

\bigskip

Besides potentially difficult simulation of diffusion bridges, there is another well known problem related to MCMC-algorithm for the problem considered. In case there are unknown parameters in the diffusion coefficient $\si$, any MCMC-scheme that includes the latent diffusion bridges leads to a scheme that is reducible. The reason for this is that a continuous sample path fixes the diffusion coefficient by means of its quadratic variation process. This phenomenon was first discussed in \cite{RobertsStramer} and a solution to it was proposed in both \cite{ChibPittShephard} and  \cite{MR2422763} within the projection-simulation setup. The resulting algorithm is referred to as the innovation scheme, as the innovations of the bridges are used as auxiliary data, instead of the discretised bridges themselves. A slightly more general solution was recently put forward in \cite{Schauer2} using the simulation-projection setup.

\subsection{Approach}\label{subsec:approach}
Assume without loss of generality that $n$ is even. The basic idea of our algorithm consists of iterating steps $(2)-(4)$ of the following algorithm:
\begin{enumerate}
\item Initialise $X_{(0:n)}$ and $\th$. 
\item For $i=1,\ldots, n/2$, sample filtered diffusion bridges $X_{(2i-2: 2i)}$, conditional on $X_{2i-2}$, $V_{2i-1}$, $X_{2i}$  and $\th$. 
\item Sample $\th$ conditional on  $X_{(0:n)}$.
\item For $i=1,\ldots, n/2-1$, sample filtered diffusion bridges $X_{(2i-1: 2i+1)}$, conditional on $X_{2i-1}$, $V_{2i}$, $X_{2i+1}$  and $\th$. Sample $X_{(0:1)}$ conditional on $V_0$, $X_1$ and $\th$. Sample $X_{(n-1:n)}$ conditional on $V_n$, $X_{n-1}$ and $\th$. 
\item Sample $\th$ conditional on  $X_{(0:n)}$.
\end{enumerate}
 
Steps (2) and (4) boil down to sampling independent bridges of the type depicted in figure \ref{fig:filteredbridge}. 
\begin{figure}
\centering
\includegraphics[scale=0.9]{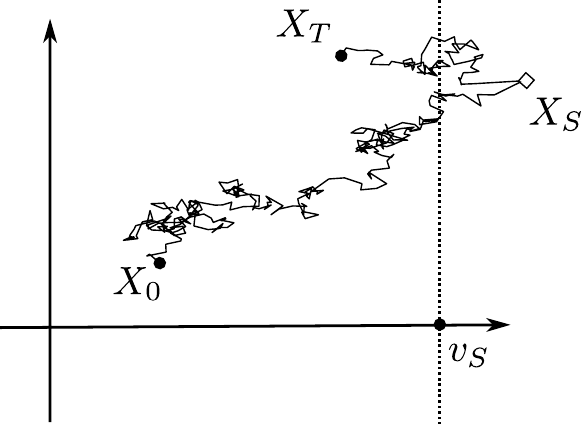}
\caption{Illustration of filtered bridges in case $L=[1\:\: 0]$ (only the first component of the diffusion is observed with error). 
Filled circles: $x_0$ and $x_T$ (fully observed). At time $S$, $v_S$ is observed; $x_S$ is unobserved. }
\label{fig:filteredbridge}
\end{figure}
Here, we have complete observations $x_0$ and $x_T$ at times $0$ and $T$ respectively, and an incomplete observation $v_S$  in between at time $S \in (0,T)$. We need to simulate a bridge connecting $x_0$ and $x_T$, while taking care of the incomplete observation at time $S$.
 For $t\in (0,S]$ this means that we  need to incorporate 2 future conditionings, an incomplete (noisy) observation at time $S$ and a complete observation at time $T$. As $X$ is Markov and we have a full observation at time $T$ this type of conditional process is independent of all observations after time $T$.  For $t\in (S,T)$ we need to sample a diffusion bridge connecting complete observations at times $S$ and $T$. The latter case has been researched in many papers over the past 15 years. See for instance  \cite{Eraker}, \cite{ElerianChibShephard}, \cite{DurhamGallant}, \cite{LinChenMykland}, \cite{BeskosPapaspiliopoulosRobertsFearnhead}, \cite{DelyonHu}, \cite{Schauer}, \cite{Stuart}, \cite{Bladt} and references therein. However, simulation of a bridge that is conditioned on one incomplete noisy observation ahead and one more complete  observation further ahead is clearly more difficult. We call such a bridge a filtered (diffusion) bridge. To the best knowledge of the authors, the problem of simulating such filtered bridges hasn't been studied in a continuous time setup. 

  Using the theory of initial enlargement of filtrations, we show in section \ref{sec:gp} that the filtered bridge process is a diffusion process itself with dynamics described by the stochastic differential equation
\[ \dd X^\star_t = b(t,X^\star_t) \dd t + \si(t,X^\star_t) \dd W_t + a(t,X^\star_t) r(t,X^\star_t) \dd t, \qquad X^\star_0=x_0 \]
Here, $a=\sigma \sigma^\T$ and the function $r$ depends both on the unknown transition density $p$ and the error density $q$. 
This SDE is derived by adapting  results on partially observed diffusions obtained by  \cite{Marchand}. 

As $p$ is intractable, direct simulating of filtered bridges from this SDE is infeasible. However, if we replace $p$ with the transition density $\tilde{p}$ of an auxiliary process $\tilde{X}$, then we can replace $r$ with the function $\tilde{r}$, where $\tilde{r}$ depends on $\tilde{p}$ in exactly the same way as $r$ depends on $p$. Exactly this approach was pursued in \cite{Schauer} in case of full observations. Naturally we choose the process $\tilde{X}$ to have  tractable transition densities. We concentrate on linear processes, where $\tilde{X}$ satisfies the SDE
\[ \dd \tilde{X}_t = \left(\tilde\beta(t) + \tilde{B}(t) \tilde{X}_t\right) \dd t + \tilde\sigma(t) \dd W_t. \]   Next, we can  simulate from the process $X^\circ$ defined by 
\[ \dd X^\circ_t = b(t,X^\circ_t) \dd t + \si(t,X^\circ_t) \dd W_t + a(t,X^\circ_t) \tilde{r}(t,X^\circ_t) \dd t, \qquad X^\circ_0=x_0 \]
instead of $X^\star$. Deviations of $X^\circ$ from $X^\star$ can be corrected by importance sampling or an appropriate acceptance probability in a Metropolis-Hastings algorithm, provided the laws of $X^\circ$ and $X^\star$ (considered as Borel measures on $C[0,T]$) are absolutely continuous.  Precise conditions for the required absolute continuity are derived in section \ref{sec:ac}. Comparing the forms of the SDE's for $X$ and $X^\circ$ we see that an additional {\it guiding term} appears in the drift for $X^\circ$. For this reason, similar as in \cite{Schauer}, we call realisations of $X^\circ$ {\it guided proposals}.

In section \ref{sec:est} we show how the innovation scheme of \cite{Schauer2} can be adopted   to the incompletely observed case considered here. Compared to \cite{Jensen} this scheme removes the restrictive assumption that the diffusion can be transformed to unit diffusion coefficient. As a more subtle important additional bonus, the scheme enables adapting the innovations to the proposals used for simulating bridges (for additional discussion on this topic we refer to \cite{Schauer2}).

A byproduct of our method is that we reconstruct paths from the incompletely observed diffusion process, which is often called smoothing in the literature. 
 
\subsection{Outline of this paper}
In section \ref{sec:gp} we derive the stochastic differential equation for the filtered bridge process corresponding to figure \ref{fig:filteredbridge}. Based on this expression we define guided proposals for filtered bridges.  In section \ref{sec:pull} we derive closed form expressions for the dynamics of the proposal process in case the measurement error is Gaussian. In section \ref{sec:ac} we provide sufficient conditions for absolute continuity of  the laws of the proposal process and true filtered bridge process. This is complemented with a closed form expression for the Radon-Nikodym derivative. The innovation scheme for estimation is presented in section \ref{sec:est}. 
The proofs of a couple of results are collected in the appendix. 

\subsection{Notation: derivatives}\label{sec:derivative-notation}
For $f\colon \RR^{m} \to \RR^n$ we denote by $\DD f$ the $m\times n$-matrix with element $(i,j)$ given by 
$\DD_{ij} f(x)= ({\partial f_j}/{\partial x_i})(x)$.
If  $n=1$, then   $\DD f$ is the column vector containing all partial derivatives of $f$. In this setting we write the $i$-th element of $\DD f$ as $\DD_i f(x)= ({\partial f}/{\partial x_i})(x)$ and denote $\DD^2 f=\DD(\DD f)$ so that $\DD^2_{ij} f(x) =
{\partial^2 f(x)}/({\partial x_i\partial x_j})$.
Derivatives with respect to time are always written as $\partial / \partial t$.

\section{Guided proposals for filtered bridges}\label{sec:gp}

Consider the filtered probability space $(\Omega, \scr{F}, (\scr{F}_t)_{t\ge 0}, \PP)$. Assume $(W_t)_{t\ge 0}$ is an $\scr{F}_t$-adapted Brownian motion. Let $X$ be a strong solution to  the SDE given in equation \eqref{eq:sde} on this setup. 

Throughout we assume $0 < S < T$. At times $0$ and $T$ we assume full observations $x_0\in \RR^d$ and $x_T\in \RR^d$ respectively. At time $S$ we assume to have the incomplete  observation $v_S \in \RR^m$ (with $m<d$). Assume that the $m$-dimensional random vector  $\eta$ has density  $q$.

We will shortly derive that the process $X$, conditioned on $Y=(V_S, X_T)$ is a diffusion process itself on a filtered probability space with a new filtration. To derive this result, we employ results of Jacod within the volume \cite{JeulinYor}  on ``grossissements de filtration'' (see also \cite{Jeulin}). Furthermore, we follow the line of reasoning outlined in \cite{Marchand}, where a similar type of problem is dealt with. The results we use are also nicely summarised in section 2 of \cite{Amendinger}.  Define the enlarged filtration by 
\[ \scr{G}_t = \bigcap_{\eps>0} \left(\scr{F}_{t+\eps} \vee \si(Y)\right). \]
The idea is to find the semi-martingale decomposition of the $\scr{F}_t$-Wiener process $W$ relative to $\scr{G}_t$.

Denote the law of the process $X$ started in $x$ at time $s$ by $\mathrm{P}^{(s,x)}$. 
 We assume that $X$ admits smooth transition densities such that  $\mathrm{P}^{(s,x)}(X_\tau \in \!\dd y) = p(s,x; \tau ,y) \dd y$ (with $\tau>s$).
Suppose $t \in [0,S)$. For $v_S \in \RR^m$ and $x_T \in \RR^d$ we have 
\begin{align*} \PP^{(t,x)}\left(V_S \le v_S, X_T \le x_T\right)& = \int \PP^{(S,\xi)}\left(\eta \le v_S-L\xi,  X_T \le x_T\right) p(t,x; S, \xi) \dd \xi \\ & = \int \PP(\eta \le v_S-L\xi)  \PP^{(S,\xi)}\left(X_T \le x_T\right)p(t,x; S, \xi) \dd \xi. \end{align*}
From this we find that for $t\in [0,S)$, $(V_S, X_T) \mid X_t=x$ has density \[ \int p(t,x; S, \xi) p(S, \xi; T, x_T) q(v_S-L \xi) \dd \xi \] with respect to Lebesgue measure on $\RR^{m+d}$. Similarly, for $t\in[S,T)$, $X_T \mid X_t = x$ has density $p(t,x; T, x_T)$. 
The function defined in the following definition plays a key role in the remainder. 
\begin{defn}
Suppose $0<S<T$. Define 
\[
	p(t, x; S, v_S; T, x_T) =\begin{cases} \int p(t,x; S, \xi) p(S, \xi; T, x_T) q(v_S-L \xi) \dd \xi &  \quad \text{if}\quad t<S \\ p(t,x; T, x_T)&  \quad \text{if}\quad S\le t <T \end{cases}. 
\]
For notational convenience we write $p(t,x)$ instead of $p(t,x; S, v_S; T, x_T)$, when it is clear from the context what the remaining  four arguments are. To avoid abuse of notation, a transition density is always written with all its four arguments. Define 
\begin{center}\rara{1.2}
\begin{tabular}{l}
$R(t,x) = \log p(t,x), \quad
	 r(t,x) = \DD R(t,x), \quad  H(t,x) = -\DD^2 R(t,x).$\\
\end{tabular}
\end{center}
Here $\DD$ denotes differentiation, with precise conventions outlined in section \ref{sec:derivative-notation}.

\end{defn}

\begin{lemma}\label{lem:partialbridge-sde}
For $t\in [0,T)$, the diffusion conditioned on $V_S=v_S$ and $X_T=x_T$ satisfies the SDE
\begin{equation}\label{xstar} \dd X^\star_t = b(t,X^\star_t) \dd t + \si(t,X^\star_t) \dd \bar{W}_t + a(t,X^\star_t) r(t,X^\star_t) \dd t, \qquad X^\star_0=x_0, \end{equation}
where $\bar{W}_t$ is a $\scr{G}_t$-Brownian motion. 
\end{lemma}
\begin{proof}
The proof is similar to the proof of Th\'eor\`eme 2.3.4 in \cite{Marchand}. The first step consists of proving that the process $(\bar{W}_t)_{t\ge 0}$ defined by 
\begin{equation}\label{eq:barW} \bar{W}_t = W_t - \int_0^t \si(s,X_s)^\T r(s, X_s; S, V_S; T, X_T) \dd s \end{equation}
is a $(\scr{G}_t)_{t\in [0,T)}$-Brownian motion, independent of $Y$.
For proving this, first define 
\[ k(s, x; S, v_S; T, x_T)= \si(s,x)' \DD\log p(s, x; S, V_S; T, X_T),  \]
where $\DD$ is assumed to act on the second argument of $p$.
For notational convenience we write $k(s,X_s)$ instead of $k(s, X_s; S, v_S; T, x_T)$.
 Then
\begin{equation}\label{eq:lhs} \int_0^t k(s, X_s) p(s, X_s) \dd \langle W\rangle_s =\int_0^t \si(s,X_s)' \DD\,  p(s,X_s) \dd s. \end{equation}
By It\=o's lemma
\begin{align*} \DD\, p(s,X_s) &=   b(s,X_s)'\DD\, p(s,X_s) \dd s + \si(s,X_s)' \DD\, p(s,X_s) \dd W_s \\ & \quad +\frac12\sum_{ij} a_{ij}(s,X_s) \left.  \left(\frac{\partial^2}{\partial x_i \partial x_j} p(s,x)\right)\right|_{x=X_s} \dd s. \end{align*}
Hence
\begin{equation}\label{eq:rhs}
\begin{split} \langle p(s,X_s), W_s\rangle_t &=\left\langle \int_0^t \si(s,X_s)' \DD\, p(s,X_s) \dd W_s, \int_0^t \dd W_s\right\rangle \\ & = \int_0^t \si(s,X_s)'\DD\, p(s,X_s) \dd s.
\end{split}
\end{equation}
Combining equation \eqref{eq:lhs} and \eqref{eq:rhs} gives 
\[ \int_0^t k(s, X_s) p(s, X_s) \dd \langle W\rangle_s =\langle p(s,X_s), W_s\rangle_t. \]
Th\'eor\`eme 2.1 of Jacod (1985)  implies that
\[ W_t - \int_0^t k(s, X_s; S, V_S; T, X_T) \dd \left\langle W \right\rangle_s, \quad t \in [0,T) \] 
is a $(\scr{G}_t)_{t\in [0,T)}$-martingale. By computing the quadratic variation of $\bar{W}$ it is seen that $\bar{W}$ is a $(\scr{G}_t)_{t\in [0,T)}$-Wiener process on $[0,T)$, independent of $\si(Y) \subset \scr{G}_0$. 

Multiplying both sides of equation \eqref{eq:barW} with $\si(t,X_t)$ and plugging in \eqref{eq:sde} gives 
\[ \dd X_t = b(t,X_t) \dd t + \si(t,X_t) \dd \bar{W}_t + a(t,X_t)  r(t, X_t; S, V_S; T, X_T)  \dd t. \]
Next, conditioning on $Y=(V_S, X_T)=(v_S, x_T)$ and using the independence of $\bar{W}$ and $Y$ gives the result. 
\end{proof}

This results demonstrates that the filtered bridge process is a diffusion process itself with an extra term superposed on the drift of the original diffusion process. The term $a(t,x) r(t,x)$ will  be referred to as the {\it pulling term}, as it ensures a pull of the diffusion process to have the right distributions at time $S$ and $T$. 
In case there is no measurement error, we have that for $t<S$
\[   p(t, x) = \int_{\{\xi\colon L \xi \,=\, v_S\}} p(t,x; S, \xi)\, p(S, \xi; T, x_T) \dd \xi.  \]
As the dynamics of the  bridge involve the unknown transition density of the process, it cannot be used directly for simulation purposes. For that reason, we propose to replace $p(\cdot, \cdot;\cdot, \cdot)$ with the transition density $\tilde{p}(\cdot, \cdot;\cdot, \cdot)$ of a process $\tilde{X}$ for which $\tilde{p}$ is tractable to obtain a proposal process $X^\circ$.
\begin{defn}
	{\it Guided proposals} are defined as solutions to the SDE
	\begin{equation}\label{xcirc} \dd X^\circ_t = b(t,X^\circ_t) \dd t + \si(t,X^\circ_t) \dd W_t + a(t,X^\circ_t) \tilde r(t,X^\circ_t) \dd t, \qquad X_0=u \end{equation}
Here $\tilde{r}(t,x) = \DD \log \tilde{p}(t,x)$, where
\[
	\tilde{p}(t, x) =\begin{cases} \int \tilde{p}(t,x; S, \xi) \tilde{p}(S, \xi; T, x_T) \tilde{q}(v_S-L \xi) \dd \xi &  \quad \text{if $t<S$} \\ \tilde{p}(t,x; T, x_T)&  \quad \text{if $S\le t <T$} \end{cases} 
\]
and $\tilde{q}$ is a probability density function on $\RR^m$, with $m=\dim(v_S)$. 
\end{defn}
This approach was initiated in \cite{Schauer}. We will assume throughout that   $\tilde{X}$ is  a linear process: 
\begin{equation}
\label{eq:linproc} \dd \tilde{X}_t = \tilde\beta(t) \dd t + \tilde{B}(t) \tilde{X}_t \dd t  + \tilde{\si}(t) \dd W_t. 
\end{equation}
Define $\tilde{R}(t,x) = \log \tilde{p}(t,x)$,
$\tilde{r}(t,x) = \DD \tilde{R}(t,x)$ and $\tilde{H}(t,x) = -\DD^2 \tilde{R}(t,x)$.

\subsection{Notation: diffusions and guided processes}

 We denote the laws of $X$, $X^\star$ and $X^\circ$ viewed as measures on the space $C([0,t], \RR^d)$ of continuous functions from $[0,t]$ to $\RR^d$ equipped with its Borel-$\sigma$-algebra by  $\PP_t$, $\PP^\star_t$ and $\PP^\circ_t$ respectively. 
For easy reference, the following table summaries the various processes and corresponding measures around.
\begin{center} \rara{1.2}
\begin{tabular}{|l|l|  l |}
\hline
$X$ & original, unconditioned diffusion process, defined by \eqref{eq:sde} & $\PP_t$\\
$X^\star$ & corresponding filtered bridge, conditioned on $v_S$ and $xv_T$, defined by \eqref{xstar}& $\PP^\star_t$\\
$X^\circ$ & proposal process defined by \eqref{xcirc}& $\PP^\circ_t$\\ 
$\tilde X$ & linear process defined by  \eqref{eq:linproc} whose transition densities $\tilde p$ appear & \\ &
in the definition of $X^\circ$& \\
\hline
\end{tabular}
\end{center}
The infinitesimal generator of the diffusion process $X$ is denoted by $\Ell$. 

\subsection{Pulling term induced by a linear process}\label{sec:pull}

In this section we derive closed form expressions for $\tilde{r}$ and $\tilde{H}$.  For the remainder of the this paper we make the following assumption.
\begin{ass}\label{ass:tilde-q}
$\tilde{q}$ is the density of the  $N(0,\Sigma)$ distribution.
\end{ass}
Note that this is an assumption on $\tilde{q}$ which appears in the proposal, and not on $q_i$ which is the density of the error at time $t_i$. 

We start with a recap of a few  well known results on linear processes. See for instance \cite{LiptserShiryayevI}.
Define  the fundamental matrix $\Phi(t)$ as the matrix satisfying
\[	\Phi(t) = I + \int_0^t \tilde{B}(\tau)\Phi(\tau) \dd \tau.\] Set $\Phi(t,s)=\Phi(t)\Phi(s)^{-1}$. 
For a linear process it is known that its transition density $\tilde{p}$ satisfies
\[ \tilde{p}(t,x; S, x_S) = \phi(x_S; \Phi(S,t)x + g_S(t), K_S(t))  \qquad 0\le t<S\]
with 
\begin{equation}\label{eq:def-g} g_S(t) = \int_t^S \Phi(S,\tau)\tilde\beta(\tau) \dd \tau\end{equation}
and
\begin{equation}\label{eq:def-K} K_S(t)= \int_t^S \Phi(S,\tau) \tilde{a}(\tau) \Phi(S,\tau)' \dd \tau. \end{equation}

\begin{lemma}
\label{lem:linearpull}
For $t<S$
\begin{equation} \tilde{r}(t,x)= \begin{bmatrix} L \Phi(S,t) \\ \Phi(T,t) \end{bmatrix}' U(t) \begin{bmatrix} v_S- L g_S(t) -L \Phi(S,t) x \\ v_T-g_T(t)-\Phi(T,t) x\end{bmatrix} 
\end{equation}
and 
\[ \tilde{H}(t) = \begin{bmatrix} L\Phi(S,t) \\ \Phi(S,t)\end{bmatrix}^\T U(t) \begin{bmatrix} L\Phi(S,t) \\ \Phi(S,t)\end{bmatrix}. \]
Here, 
\begin{equation}
\label{eq:U}	U(t)=\begin{bmatrix} L K_S(t) L^\T +\Sigma  & LK_S(t)
\Phi(T,S)' \\ \Phi(T,S) K_S(t) L^\T & K_T(t) \end{bmatrix} ^{-1}. 
\end{equation}
\end{lemma}
\begin{proof}
The proof is given in section \ref{subsec:proof:lem:linearpull}. 
\end{proof}

\begin{corr}
\label{corr:pull-witherror}
Assume $\tilde{a}(t)\equiv\tilde{a}$ and $\tilde{B}(t)\equiv 0$. 
Define
\begin{align}\label{eq:defN-Q_eta}	N(t) &= \left(L\tilde{a} L^\T +   \frac{T-t}{(S-t)(T-S)}\Sigma \right)^{-1}
\\
	Q(t) &= L^\T N(t) L. 
\end{align}
Then 
\begin{equation}
\label{eq:noisy-pull} 	\tilde{r}(t,x)  = \begin{cases} Q(t) \frac{h_S(t,x)}{S-t} 
 + \left\{\tilde{a}^{-1}  -Q(t)\right\}\frac{h_T(t,x)}{T-t} & \text{if  } t \in [0,S) \\ \tilde{a}^{-1} \frac{h_T(t,x)}{T-t} & \text{if  } t \in [S,T) \end{cases}
\end{equation}
and 
\[ (T-t)\tilde{H}(t)= \begin{cases}\tilde{a}^{-1} + \frac{T-S}{S-t} Q(t)& \text{if  } t \in [0,S) \\ \tilde{a}^{-1} & \text{if  } t \in [S,T) \end{cases}.  \]
Here, 
\[
h_S(t,x) = u_S-\int_t^S \tilde\beta(\tau)\dd \tau -x \quad \text{and}\quad  h_T(t,x)=x_T-\int_t^T\tilde\beta(\tau)\dd \tau -x
\]
with $u_S$ any vector such that $L u_S = v_S$.

Moreover, 
\[ \lim_{t\uparrow S} \tilde{r}(t,x) = L^\T \Sigma^{-1} L (u_S-x) + \tilde{a}^{-1} \frac{h_T(S,x)}{T-S} \]
\end{corr}
\begin{proof}
In this case we can carry out the inversion in equation \eqref{eq:U} in closed form. The proof is given in section \ref{sec:corr:pull-witherror}.
\end{proof}  
\begin{rem}
Suppose $L=I_{d\times d}$ and $\Sigma=0_{d\times d}$ which corresponds to a full observation at time $S$ without error. Then $Q(t)=\tilde{a}^{-1}$ and the second term in $\tilde{r}$ (for $t<S$) disappears. 
Furthermore then, $\tilde{H}(t)=\tilde{a}^{-1} (S-t)^{-1}$.  In this way, we recover the result for the full observation case. 
\end{rem}

\begin{ex}  Suppose $X_t$ is a two-dimensional Brownian Motion, where we only observe the first component at time $S$ and both at time $T$. In this case $L=[1, 0]$,   $g\equiv 0$ and $\Phi=I_{2\times 2}$. It is easy to see that 
\[	N(t) = \left(1 + \frac{T-t}{(S-t)(T-S)} \Sigma\right)^{-1}=\frac{(S-t)(T-S)}{(S-t)(T-S)+\Sigma (T-t)}. \]
By corollary \ref{corr:pull-witherror}, it follows that for $t<S$
\[	r(t,x)=\begin{bmatrix} 1  \\ 0 \end{bmatrix} N(t) \frac{v_S-L x}{S-t} + \frac{x_T-x}{T-t}- \begin{bmatrix} 1 \\ 0 \end{bmatrix} N(t)  \begin{bmatrix} 1  & 0 \end{bmatrix} \frac{v_T-x}{T-t}.
\]
Denote the $i$-th component of a vector $x$ by $x^{(i)}$. 
The first component of $r$ equals
\[ N(t) \frac{v_S-x^{(1)}}{S-t} +(1-N(t)) \frac{x_T^{(1)}-x^{(1)}}{T-t}, \]
while the second component equals $(T-t)^{-1} \left(x_T^{(2)}-x^{(2)}\right)$. From this, we see that the second component is the same as when there would be no conditioning at time $S$. 
\end{ex}

\section{Absolute continuity result}\label{sec:ac}

In this section we derive conditions for which $\PP^*_T \ll \PP^\circ_T$ and give a closed form expression for the Radon-Nikodym derivative. 
We have the following assumption on $X$.
\begin{ass}\label{ass:X}
\begin{enumerate}
\item The functions $b$ and $\si$ are uniformly bounded, Lipschitz in both arguments and satisfy a linear growth condition on their second argument.
\item   Kolmogorov's backward equation holds: 
\[ \frac{\partial}{\partial s} p(s,x; t,y) = (\scr{L} p)(s,x; t,y)=0. \]
Here $\scr{L}$  acts on $(s, x)$.
\item Uniform ellipticity: there exists an $\eps>0$ such that for all $s\in [0,T]$, $x\in \RR^d$ and $y\in \RR^d$
\[ y^\T a(s,x) y \ge \eps \|y\|^2. \]
\end{enumerate}
\end{ass}
We have the following assumption on $\tilde X$.
\begin{ass}\label{ass:Xtilde}
$\tilde{B}$ and $\tilde\beta$ are continuously differentiable on $[S,T]$, $\tilde\sigma$ is Lipschitz on $[S,T]$ and there exists a $\tilde\eps>0$ such that for all $t\in [S, T]$ and all $y\in \RR^d$,
\[ y^\T \tilde{a}(t) y \ge \tilde\eps \|y\|^2. \] 
 \end{ass}
 
\begin{thm}\label{thm:equivalence}
Suppose  assumptions \ref{ass:X} and \ref{ass:Xtilde} apply.
Define 
\begin{equation}
\label{eq:defpsi} \Psi(X^\circ;t)=\exp\left(\int_0^t G(s,X^\circ_s) \dd s \right), \quad t < T, 
\end{equation}
where
\begin{align}\label{eq:G} G(s,x) &= (b(s,x) - \tilde b(s,x))^\T \tilde r(s,x) \nonumber \\ & \qquad -  \frac12 \trace\left(\left[a(s, x) - \tilde a(s, x)\right] \left[\tilde H(s,x)-\tilde{r}(s,x)\tilde{r}(s,x)^\T\right]\right).
\end{align}
 
 If $\tilde{a}(T)=a(T,x_T)$, then $X^\star$ and $X^\circ$ are equivalent on $[0,T]$ with Radon-Nikodym derivative given by 
\[	\frac{\dd\PP^\star_T}{\dd\PP^\circ_T} (X^\circ) = \frac{\tilde{p}(0,u)}{p(0,u)} \frac{q(v_S-LX^\circ_S)}{\tilde{q}(v_S-LX^\circ_S)} \Psi(X^\circ;T).\]
 
\end{thm}
The proof is given in the next subsection. 

\begin{rem}
In case there is no measurement error, we conjecture that absolute continuity will hold provided one takes $\tilde{a}$ such that $\tilde{a}(T)=a(T,x_T)$ and $L\tilde{a}(S) L^\T = L a(S, x_S) L^\T$.  Such a choice is possible if 
$ L a(S, x_S) L^\T$ only depends on $S$ and $L x_S$ (and not on unobserved parts of $x_S$). 
\end{rem}


\subsection{Proof of theorem \ref{thm:equivalence}}
For proving theorem \ref{thm:equivalence}, we need a few intermediate results. 
\begin{lemma}
	\label{lem:harmonic}
If we define the process $(Z_t,\, t\in [0,T))$  by $Z_t = p(t,X_t)$, then $(Z_t)$ is a $\scr{F}_t$-martingale.
\end{lemma}
\begin{proof}
	For $0\le s \le t \le S$,
	\begin{align*} \EE [Z_t \mid \scr{F}_s]=& \EE \left[p(t, X_t)\mid \mathcal F_s\right]\\& =   \int \left( \int p(t,x; S, \xi) p(S, \xi; T, x_T) q(v_S-L \xi) \dd \xi \right) p(s, X_s; t, x)   \dd x\\ &
	=\int \left(\int p(s, X_s; t, x)   p(t,x; S, \xi) \dd x \right)    p(S, \xi; T, x_T) q(v_S-L \xi) \dd \xi  \\ &
	= \int p(s, X_s, S, \xi) p(S, \xi; T, x_T)  q(v_S-L \xi)\dd \xi 
	\\ &
	= p(s, X_s)=Z_s,
	\end{align*}
where we applied the Markov property at the second equality, Fubini at the third equality and the Chapman-Kolmogorov equations at the fourth equality. 
The argument on $[S,T)$ follows along the same lines. 
\end{proof}

\begin{corr}\label{cor:backward}
The function $p(t,x)$ satisfies  Kolmogorov's backward equation  both for $t\in (0,S)$ and $t\in (S,T)$:
\[	\frac{\partial}{\partial t} p(t,x) + \scr{L} (p(t,x))=0. \]
\end{corr}

\begin{proof}
The generator of the space-time process $(t,X_t)$ is given by $\scr{K}=(\partial / \partial t) + \scr{L}$. 
As $p(t,X_t)$ is a martingale, $(t,x) \mapsto p(t,x)$ is space-time harmonic: $\scr{K} p(t,x)=0$ (Cf.\ proposition 1.7 of chapter VII in \cite{RevuzYor}). This is exactly Kolmogorov's backward equation. 
\end{proof}
\begin{lemma}\label{lem:limitp}
\[  \lim_{t \uparrow S} \frac{p(S,x)}{p(t,x)} = \frac1{q(Lx-v_S)}. \]
and similarly for $\tilde{p}$ (with $\tilde{q}$ appearing in the limit). 
\end{lemma} 
\begin{proof}
First note that under our assumptions on $b$ and $\si$,  theorem 21.11 in \cite{Kallenberg} implies that  the process $X$ is Feller. 
Take $t<S$. The transition operator is defined by 
\[	P_{t,S} f(x) = \int p(t,x; S, \xi) f(\xi) \dd \xi. \]
Hence with $f(\xi)= p(S, \xi; T, x_T) q(L \xi - v_S)$
\[	p(t, x) = \int p(t,x; S, \xi) p(S, \xi; T, x_T) q(v_S-L \xi) \dd \xi = P_{t,S} f(x) 
\]
As $X$ is Feller, $\lim_{t\uparrow S} P_{t,S} f(x) = f(x)$ from which the result follows easily.    
\end{proof}

\begin{lemma}
 \label{lem:properties-rtilde}
Suppose $t \in [S,T)$. Then $\tilde{r}$ is Lipschitz in its second argument and satisfies a linear growth condition on both $[0,S)$ and $[S,t]$. 
\end{lemma}
\begin{proof}
On $[0,S)$, it is clear from lemma \eqref{lem:linearpull} that $x\mapsto \tilde{r}(t,x)$ is linear. On $[S,t]$ this is proved in \cite{Schauer}. 
\end{proof}

\begin{proof}[Proof of theorem \ref{thm:equivalence}]
The proof follows   the line of proof in proposition 1 of \cite{Schauer}. Consider $t \in [S,T)$. By lemma \ref{lem:properties-rtilde}, $\tilde{r}$ is Lipschitz in its second argument and satisfies a linear growth condition on both $[0,S)$ and $[S,t]$. Hence, a unique strong solution of the SDE for $X^\circ$ exists on $[0,t]$.

By Girsanov's theorem (see e.g.\ \cite{LiptserShiryayevI})
the laws of the processes $X$ and $X^\circ$ on $[0,t]$ are equivalent and the corresponding 
Radon-Nikodym derivative is given by 
\[
\frac{\dd\PP_{t}}{\dd\PP^\circ_{t}}({X^\circ})= 
\exp\Big(\int_0^{t} \ga_s^\T  \dd W_s  - \frac12 \int_0^{t} \|\ga_s\|^2 \dd s\Big), 
\]
where $W$ is a Brownian motion under $\PP^\circ_{t}$  and $\ga_s = \ga(s, X^\circ_s)$ solves
\[
 \sigma(s, X^\circ_s) \ga(s, X^\circ_s)  =  b(s, X^\circ_s)-b^\circ(s, X^\circ_s).
\]
(Here we  lightened notation by writing $\ga_s$ instead of $\ga(s, X^\circ_s)$. In the 
remainder of the proof we follow the same convention and apply it to  other processes as well.) 
Observe that by definition of $\tilde r$ and $b^\circ$ we have
$\ga_s =  -\sigma^\T_s \tilde{r}_s$ and $\|\beta_s\|^2 = \tilde{r}^\T_s a_s \tilde{r}_s$, hence 
\begin{equation}\label{eq:girs1}
\frac{\dd\PP_{t}}{\dd\PP^\circ_{t}}({X^\circ})= 
\exp\Big(-\int_0^{t} \tilde{r}_s^\T \sigma_s  \dd W_s  - \frac12 \int_0^{t} \tilde{r}^\T_s a_s \tilde{r}_s \dd s\Big). 
\end{equation}
Denote the infinitesimal operator of $X^\circ$ by $\Ell^\circ$. By definition of $X^\circ$ 
and $\tilde R$ we have $\Ell^\circ \tilde R = \Ell \tilde R + \tilde{r}^\T a \tilde{r}$.
By   It\=o's formula \[
\tilde R_t - \tilde R_S = \int_{[S,t)} \Big(\frac{\partial}{\partial s}\tilde R_s + \Ell \tilde R_s\Big)\,\dd s + 
\int_{[S,t)} \tilde{r}^\T_s a_s \tilde{r}_s\,\dd s 
+ \int_{[S,t)}  \tilde{r}_s^\T\sigma_s\,\dd W_s.
\]
Applying It\=o's formula in exactly the same manner on $[0,s]$ with $s<S$ and subsequently taking the limit $s\uparrow S$ we get 
\[
\tilde R_{S-} - \tilde R_0 = \int_{[0,S)} \Big(\frac{\partial}{\partial s}\tilde R_s + \Ell \tilde R_s\Big)\,\dd s + 
\int_{[0,S)} \tilde{r}^\T_s a_s \tilde{r}_s\,\dd s 
+ \int_{[0,S)}  \tilde{r}_s^\T\sigma_s\,\dd W_s.
\]
Combining the preceding two displays with \eqref{eq:girs1} we get 
\begin{equation}\label{eq	:PP-PPcirc}
\frac{\dd\PP_{t}}{\dd\PP^\circ_{t}}({X^\circ})=
\exp\left(-\tilde R_{t} + \tilde R_0+
 \tilde R_S - \tilde R_{S-}+\int_0^t G_s\dd s\right), 
\end{equation}
where
\begin{equation}\label{eq:Gprelim}
	G=\left( \frac{\partial}{\partial s} \tilde R_s+ \Ell \tilde R_s\right) +\frac12\tilde{r}_s^\T a  \tilde{r}_s .
\end{equation}
If $p(t,x)$ and $\tilde{p}(t,x)$ satisfy Kolmogorov's backward equation, then the first  term between brackets  on the right-hand-side of this display equals
$ \Ell \tilde{R}-\tilde\Ell \tilde{R} -\frac12 \tilde{r}^\T\tilde{a} \tilde{r}$.
This follows from lemma 1 in \cite{Schauer}. 
This is naturally the case on $(S,T)$ and by corollary \ref{cor:backward} on  $(0,S)$ as well. 
Substituting this in equation \ref{eq:Gprelim} we arrive at the expression for $G$  as given in the statement of the theorem.  By lemma \ref{lem:limitp}
\begin{align*} -\tilde R_{t} + \tilde R_0+
 \tilde R_S - \tilde R_{S-}&= \log \left(\frac{\tilde{p}(0,u)}{\tilde{p}(t,X^\circ_t)} \frac{\tilde{p}(S, X^\circ_S)}{\tilde{p}(S-, X^\circ_{S-})}\right) \\ & = \log \left(\frac{\tilde{p}(0,u)}{\tilde{p}(t,X^\circ_t)} \frac1{\tilde{q}(v_S-L X^\circ_S )}\right).
\end{align*}
Combined with equation \eqref{eq	:PP-PPcirc}, we obtain
\[ \frac{\dd\PP_{t}}{\dd\PP^\circ_{t}}({X^\circ})= \frac{\tilde{p}(0,u)}{\tilde{p}(t,X^\circ_t)} \frac1{\tilde{q}(v_S-L X^\circ_S)} \exp\left(\int_0^t G(s, X^\circ_s) \dd s \right). \]
As entirely similar calculation reveals that 
\begin{equation*}
\frac{\dd\PP^\star_{t}}{\dd\PP_t}(X) = \frac{p(t,X_t)}{p(0,u)}\frac{p(S-, X_{S-})}{p(S, X_S)}= \frac{p(t,X_t)}{p(0,u)} q(v_S-LX_S).
\end{equation*}
Combining the previous two displays gives
\begin{align*}	\frac{\dd\PP^\star_t}{\dd\PP^\circ_t} (X^\circ) &= \frac{\tilde{p}(0,u)}{p(0,u)} \frac{p(t, X^\circ_t)}{\tilde{p}(t, X^\circ_t)} \frac{q(v_S-L X^\circ_S)}{\tilde{q}(v_S-L X^\circ_S)} \Psi(X^\circ, t)\\
& = \frac{\tilde{p}(0,u)}{p(0,u)} \frac{q(v_S-L X^\circ_S)}{\tilde{q}(v_S-L X^\circ_S)} \frac{p(t, X^\circ_t; T, v_T)}{\tilde{p}(t, X^\circ_t;  T, v_T)} \Psi(X^\circ; t).
\end{align*}
From here, the limiting argument $t\uparrow T$ is exactly as in  \cite{Schauer}.
\end{proof}

\section{Special bridges near $t_0$ and $t_n$}\label{sec:boundary}

In section \ref{sec:est}  we will need 
 filtered processes which take the boundary conditions near $t_0$ and $t_n$ into account (besides the filtered bridge introduced above). 

\subsection{Near the endpoint $t_n$}

Near $t_n$ we wish to simulate a filtered bridge conditioned on $X_{n-1}$ and $V_n$ on $[t_{n-1}, t_n]$. For this purpose, we derive the dynamics of a diffusion process starting in $X_0 = x_0$, conditioned  $V_S=LX_S+\eta$.  We can use exactly the same techniques as in sections \ref{sec:gp} and \ref{sec:pull} to derive the SDE for the conditioned process. In this case, $p(t,x; S, v_S; T, x_T)$ should be replaced by
\[ p_{\en}(t,x) := \int p(t,x; S,\xi) q(v_S-L\xi) \dd \xi. \]
In lemma \ref{lem:linearpull} we should replace $\tilde{r}(t,x)$ and $\tilde{H}(t)$ by
\begin{equation}\label{eq:rtilde} \tilde{r}_{\en}(t,x)=\Phi(S,t)^\T L^\T\left(L K_S(t) L^\T + \Sigma\right)^{-1} L\left(u_S- g_S(t)-\Psi(S,t) x\right) \end{equation}
and 
\[   \tilde{H}_{\en}(t) = \Phi(S,t)^\T L^\T\left(L K_S(t) L^\T + \Sigma\right)^{-1}L\Psi(S,t)\]
respectively. 
Then $X^\star$ and $X^\circ$ are equivalent on $[0,S]$ with Radon-Nikodym derivative given by 
\[	\frac{\dd\PP^\star_S}{\dd\PP^\circ_S} (X^\circ) = \frac{\tilde{p}_\en(0,u)}{p_\en(0,u)} \frac{q(v_S-LX^\circ_S)}{\tilde{q}(v_S-LX^\circ_S)} \Psi(X^\circ;S).\]


\subsection{Near the starting point $t_0$}

Near $t_0$ we wish to simulate a filtered bridge conditioned on $V_0$ and $X_{1}$ on $[t_0, t_1]$.   Assume $X_0$ has prior distribution $\nu$.  We simulate the filtered bridge in two steps:
\begin{enumerate}
\item  simulate $X_0$, conditional on $(v_0, x_S)$;
\item simulate a  bridge  connecting $x_0$ (the realisation of $X_0$) and $x_S$.
\end{enumerate}
Suppose we wish to update $(x_0, X^\circ)$ to $(\bar{x}_0, \bar{X}^\circ)$ (the proposal). Each proposal will be generated by first drawing $\bar{x}_0$ conditional on $x_0$ using some kernel $q(\bar{x}_0 \mid x_0)$ followed by  sampling a bridge connecting $x_0$ and $x_S$. Denote the conditional density of $x_0$ conditional on $v_0$ by $\nu(x_0 \mid v_0)$.  The ``target density'' is proportional to
\[ \frac{\dd \PP_S^\star}{\dd \PP_S^\circ}(X^\circ) p(0, x_0; S, x_S) \nu(x_0 \mid v_0) = \frac{\Psi(X^\circ,S) \nu(x_0 \mid v_0)}{\tilde{p}(0,x_0; S, v_S)}  \]
(note that the intractable term $p(0, x_0; S, x_S)$ cancels). 
The acceptance probability then equals $A\wedge 1$, where
\[ A= \frac{\Psi(\bar{X}^\circ,S)}{\Psi(X^\circ,S)} \frac{\nu(\bar{x}_0 \mid v_0)}{\nu(x_0 \mid v_0)}    \frac{\tilde{p}(0,\bar{x}_0; S, v_S)}{\tilde{p}(0,x_0; S, v_S)}    \frac{q(x_0 \mid \bar{x}_0)}{q(\bar{x}_0 \mid x_0)}.\]
When $\eta$ (the distribution of the noise on the observations) is  $N(0,\Sigma)$, a tractable expression for $\nu(x_0 \mid \nu_0)$ is obtained by taking $\nu \sim N(\mu,C)$. In that case
the vector $[x_0, v_0]$ is jointly Gaussian which implies that \[	X_0 \mid V_0=v_0 \sim N\left(\mu+CL^\T (LCL^\T + \Sigma)^{-1} (v_0-L\mu), C-   CL^\T (LCL^\T + \Sigma)^{-1} L C\right). \]

\section{Estimation by MCMC using temporary reparametrisation}\label{sec:est}
 
In this section we present a novel algorithm to draw from the posterior of $\th$ based on incomplete observations. 
The basic idea for this algorithm is quite simple and outlined in section \ref{subsec:approach}. Unfortunately, this basic scheme collapses in case there are unknown parameters in the diffusion coefficient. This is a well known phenomenon when applying data-augmentation for estimation of discretely observed diffusions. It was first noticed by \cite{RobertsStramer} and we refer to that paper for a detailed explanation. \cite{MR2422763} developed an MCMC algorithm that alternatively updates the parameter and the driving Brownian motion increments of the proposal process. Their derivation was developed entirely by first discretising the process. \cite{Schauer2} showed how this algorithm can be derived in the simulation-projection setup. Quoting from this paper: ``The basic idea is that  the laws of the bridge proposals can be understood as parametrised push forwards of the law of an underlying random process common to all models with different parameters $\th$. This is naturally the case for proposals defined as solutions of stochastic differential equations and the driving Brownian motion can be taken as such underlying random process.''

Here, we propose to derive such an algorithm in  case of incomplete observations, which complicates the derivations considerably. We define a Metropolis-Hastings algorithm that uses temporary reparametrisations. Suppose $t\in (a,b)$ and let $Z$ be a continuous stochastic process. Let $Z_{(a,b)}=(Z_t,\, t\in (a,b))$. Let $s\in (a,b)$. 
Define $X^\star_{(a,b)}=(X^\star_t,\, t\in (a,b))$ as the solution to the SDE
\[ \!\dd X^\star_t =\left( b(t, X^\star_t) + a(t, X^\star_t) \tilde{r}_\th(t, X^\star_t; s, V_s; b, x_b)\right) \dd t + \si(t, X^\star_t) \dd Z_t,\qquad X^\star_a= x_a. \]
{\it Assume $\si$ is invertible. } We define the mapping $g_{(x_a, v_s, x_b)}$ that maps $(\th, Z_{(a,b)})$ to $(\th, X^\star_{(a,b)})$ and  define an inverse mapping $g^{-1}_{(x_a, v_s, x_b)}$ that maps $(\th, X^\star_{(a,b)})$ to $(\th, Z_{(a,b)})$. The process $Z_{(a,b)}$ is referred to as the {\it innovation process}. The main idea of the algorithm below is that when we update $(\th, X^\star)$ in blocks, we temporarily reparametrise to $(\th, Z)$.

In the algorithm below, we assume $n$ is even (adaptation to the case where $n$ is odd is straightforward).
 For $i<j$ denote $Z_{(i:j)}=\{Z_t, t\in (t_i, t_j)\}$ and $X^\star_{(i:j)}=\{X^\star_t, t\in (t_i, t_j)\}$.

We refer to subsection \ref{sec:boundary} for  simulation of $X^\star_{(0:1)}$ and $X^\star_{(n-1:n)}$ at the boundaries. We  write $g_{(V_0, X_1)}$ for the corresponding map from $(\theta, Z_{(0:1)})$ to $(\theta, X^\circ_{(0:1)})$ and similarly $g_{(X_{n-1}, V_n)}$
for the map from $(\theta, Z_{(n-1:n)})$ to $(\theta, X^\circ_{(n-1:n)})$. 
In order to conveniently handle boundary cases in the algorithm below we make the convention that the expressions $(-1: 1)$ and $g_{(X_{-1},V_0, X_{1})}$ are to be understood as  $(t_0,  t_1)$ and  $g_{(V_0, X_{1})}$ respectively. We use a similar convention on the right boundary. 

Define 
\[  X^\star_\even= \{X^\star_{2i},\, i=0,\ldots, n/2\} \qquad X^\star_\odd= \{X^\star_{2i+1},\, i=0,\ldots, n/2-1\}. \]
We change the notation on $\Psi$ defined in \eqref{eq:defpsi}  slightly to accommodate dependence on $\th$:
\[ \Psi\left((\th, Z_{(a:b)})\right) =\exp\left(\int_a^b G_\th(t, g_{(x_a, v_s, x_b)}(\theta, Z_{(a:b)})(t)) \dd t\right) \]
with the modifications for the boundary cases as before.

We propose the following algorithm.

\begin{algorithm}\label{alg1} \

\begin{enumerate}
\item {\bf Initialisation.} Choose a starting value for $\th$ and initialise $X^\star_{[0,T]}$. 

\item {\bf Update} $\{Z_{(2i-2:2i)},\, i\in \scr{I}\} \mid (\th, \scr{D}, X^\star_\even)$.  Independently, for $i=1,\ldots, n/2$ do
\begin{enumerate}
\item Compute $Z_{(2i-2:2i)}=g^{-1}_{(X^\star_{2i-2}, V_{2i-1}, X^\star_{2i})}(\th, X^\star_{(2i-2:2i)})$.
\item  Sample a Wiener process $Z^\circ_{(2i-2:2i)}$.
\item  Sample $U \sim \scr{U}(0,1)$. Compute
\[ A_1 =\frac{\Psi\left(g_{(X^\star_{2i-2}, V_{2i-1}, X^\star_{2i})}(\th,Z^\circ_{(2i-2:2i)})\right)}{\Psi\left(g_{(X^\star_{2i-2}, V_{2i-1}, X^\star_{2i})}(\th, Z_{(2i-2:2i)})\right)}. \]
Set 
\[ Z_{(2i-2:2i)} := \begin{cases} Z^\circ_{(2i-2:2i)}& \text{if}\quad U\le A_1 \\
 Z_{(2i-2:2i)} & \text{if}\quad U> A_1 \end{cases}. \]
\end{enumerate}

\item {\bf Update} $\th \mid (\{Z_{(2i-2:2i)},\, i=1,\ldots, n/2\}, \scr{D}, X^\star_{\even})$. 
\begin{enumerate}
\item  Sample ${\th^\circ} \sim q(\cdot \mid \th)$.
\item  Sample $U \sim \scr{U}(0,1)$. Compute
\begin{align*}
A_2= \frac{\pi_0({\th^\circ})}{\pi_0(\th)} &\frac{q(\th \mid {\th^\circ})}{q({\th^\circ} \mid \theta)} \prod_{i=1}^{n/2} \left[\frac{\tilde{p}_{\th_0}(t_{2i-2}, X^\star_{2i-2}; t_{2i-1}, V_{2i-1}; t_{2i}, X^\star_{2i})}{\tilde{p}_{\th}(t_{2i-2}, X^\star_{2i-2}; t_{2i-1}, V_{2i-1};t_{2i}, X^\star_{2i})} \right.\\& \left.\times\frac{q(V_{2i-1}-L_{2i-1}X^\star_{2i-1})}{\tilde{q}(V_{2i-1}-L_{2i-1}X^\star_{2i-1})}
\frac{\Psi\left(g_{(X^\star_{2i-2}, V_{2i-1}, X^\star_{2i})}(\th^\circ, Z_{2i-2:2i}^\star)\right)}{\Psi\left(g_{(X^\star_{2i-2}, V_{2i-1}, X^\star_{2i})}(\th, Z_{2i-2:2i}^\star)\right)} \right]. 
\end{align*}
Set 
\[\th := \begin{cases} {\th^\circ} & \text{if}\quad U\le A_2 \\
 \th & \text{if}\quad U> A_2 \end{cases}. \]
\end{enumerate}

\item {\bf Adjust $X^\star$.}  For $i =1,\ldots, n/2$ compute \[X^\star_{(2i-2:2i)}=g_{(X^\star_{2i-2}, V_{2i-1}, X^\star_{2i})}(\th, Z_{(2i-2:2i)}).\] 

\item  {\bf Update} $\{Z_{(2i-1:2i+1)},\, i=0,\ldots, n/2\} \mid (\th, \scr{D}, X^\star_\odd)$.  
 Independently, for $i=0,\ldots, n/2$ do
 
\begin{enumerate}
\item Compute $Z_{(2i-1:2i+1)}=g^{-1}_{(X^\star_{2i-1}, V_{2i}, X^\star_{2i+1})}(\th, X^\star_{(2i-1:2i+1)})$.
\item  Sample a Wiener process $Z^\circ_{(2i-1:2i+1)}$.
\item  Sample $U \sim \scr{U}(0,1)$. Compute
\[ A_3 =\frac{\Psi\left(g_{(X^\star_{2i-1}, V_{2i}, X^\star_{2i+1})}(\th,Z^\circ_{(2i-1:2i+1)})\right)}{\Psi\left(g_{(X^\star_{2i-1}, V_{2i}, X^\star_{2i+1})}(\th, Z_{(2i-1:2i+1)})\right)}. \]
Set 
\[ Z_{(2i-1:2i+1)} := \begin{cases} Z^\circ_{(2i-1:2i+1)}& \text{if}\quad U\le A_3 \\
 Z_{(2i-1:2i+1)} & \text{if}\quad U> A_3 \end{cases}. \]
\end{enumerate}

\item {\bf Update} $\th \mid (\{Z_{(2i-1:2i+1)},\, i=0,\ldots, n/2\}, \scr{D}, {X}^\star_{\odd})$. 
\begin{enumerate}
\item  Sample ${\th^\circ} \sim q(\cdot \mid \th)$.
\item  Sample $U \sim \scr{U}(0,1)$. Compute
\begin{align*}
A_4= \frac{\pi_0({\th^\circ})}{\pi_0(\th)} &\frac{q(\th \mid {\th^\circ})}{q({\th^\circ} \mid \theta)}
 \prod_{i=1}^{n/2-1} \left[\frac{\tilde{p}_{\th_0}(t_{2i-1}, X^\star_{2i-1}; t_{2i}, V_{2i}; t_{2i+1}, X^\star_{2i+1})}{\tilde{p}_{\th}(t_{2i-1}, X^\star_{2i-1}; t_{2i}, V_{2i};t_{2i+1}, X^\star_{2i+1})} \right.\\& \qquad \left.\times\frac{q(V_{2i}-L_{2i}X^\star_{2i})}{\tilde{q}(V_{2i}-L_{2i}X^\star_{2i})}
\frac{\Psi\left(g_{(X^\star_{2i-1}, V_{2i}, X^\star_{2i+1})}(\th^\circ, Z_{2i-1:2i+1}^\star)\right)}{\Psi\left(g_{(X^\star_{2i-1}, V_{2i}, X^\star_{2i+1})}(\th, Z_{2i-1:2i+1}^\star)\right)} \right]..
	\end{align*}
Set 
\[\th := \begin{cases} {\th^\circ} & \text{if}\quad U\le A_4 \\
 \th & \text{if}\quad U> A_4 \end{cases}. \]
\end{enumerate}

\item {\bf Adjust $X^\star$.}  For $i =1,\ldots, n/2$ compute \[X^\star_{(2i-1:2i+1)}=g_{(X^\star_{2i-1}, V_{2i}, X^\star_{2i+1})}(\th, Z_{(2i-1:2i+1)}).\]

\item Repeat steps (2)--(7).
\end{enumerate}
\end{algorithm}

The parameter $\th$ gets updated twice during a full cycle of the algorithm, but one can choose to either omit step (3) or (6).  
 The proof that $A_2$ and $A_4$ are the correct acceptance probabilities goes along the same lines as in the completely observed case discussed in  \cite{Schauer2}. 
As demonstrated in there, in steps {\it 2(a)} and {\it 5(a)} , one can also propose $Z^\circ$ based on the current value of $Z$ in the following way
\[ Z^\circ_t = \sqrt\rho Z_t + \sqrt{1-\rho} W_t, \]
where $\rho\in [0,1)$ and $W$ is a Wiener process  that is independent of $Z$. The acceptance probability remains the same under this proposal.

\begin{rem}
If $q_i$ (the density of the noise at time $t_i$) depends on an unknown parameter $\eps$,  then we equip this parameter with a prior density $\pi_0(\epsilon)$. The parameter $\epsilon$ can then be updated in a straightforward manner in a separate Metropolis-Hastings step given the full path and the observations.
\end{rem}



\section{Proofs and Lemmas}


\subsection{Proof of lemma \ref{lem:linearpull}}
\label{subsec:proof:lem:linearpull}

For notational convenience we sometimes drop dependence on $t$. For instance, we may write $U$ instead of $U(t)$. 

To compute the pulling term at time $t$ we need to obtain to density of $(L X_S +\eta, X_T)$ conditional on $X_t$. First, we obtain the density of $(X_S, X_T) \mid X_t$. For this, note that their joint density is given by 
$\tilde{p}(t,x; S, x_S) \tilde{p}(S, x_S; T, x_T)$. 
Hence, 
\begin{align*}
&\tilde{p}(t,x; S, x_S) \tilde{p}(S, x_S; T, x_T) \propto\\& \qquad  \exp\left( - \frac12 \left(x_S-g_S(t)-\Phi(S,t)x\right)^\prime K_S(t)^{-1}\left(x_S-g_S(t)-\Phi(S,t)x\right)\right)\\ & \qquad  \times
   \exp\left( - \frac12 \left(x_T-g_T(S)-\Phi(T,S) x_S\right)^\prime K_T(S)^{-1} \left(x_T-g_T(S)-\Phi(T,S) x_S\right) \right).
\end{align*}
The exponent equals 
\[ -\frac12 \begin{bmatrix} x^\prime_S & x^\prime_T \end{bmatrix} A  \begin{bmatrix} x_S \\ x_T \end{bmatrix} + q^\prime \begin{bmatrix} x_S \\ x_T \end{bmatrix} \]
with
\begin{align*}
 A_{11} &= K_S(t)^{-1} + \Phi(T,S)^\prime K_T(S)^{-1} \Phi(T,S)\\
A_{12} & = -\Phi(T,S)^\prime K_T(S)^{-1}  \\
A_{21} & = A_{12}^\prime \\
A_{22} & = K_T(S)^{-1} \\
q_1 & = 	K_S(t)^{-1}(g_S(t)+\Phi(S,t)x_t)  - \Phi(T,S)^\prime K_T(S)^{-1} g_T(S) \\
q_2 & =  K_T(S)^{-1} g_T(S)
\end{align*}
This implies that the joint distribution of $(X_S, X_T)$ conditional on $X_t$ is normal with covariance matrix $\Upsilon=A^{-1}$ and mean vector $\mu=\Upsilon q$. Here, (using expressions for the inverse of a partitioned matrix and Woodbury's formula)
\begin{align*} \Upsilon = \begin{bmatrix} \Upsilon_{11} & \Upsilon_{12} \\ \Upsilon_{12}^\prime & \Upsilon_{22}   \end{bmatrix} &=\begin{bmatrix} K_S(t) & K_S(t) \Phi(T,S)^\prime \\ \Phi(T,S) K_S(t) & K_T(S)+\Phi(T,S) K_S(t) \Phi(T,S)^\prime \end{bmatrix} \\ &=\begin{bmatrix} K_S(t) & K_S(t)
\Phi(T,S)^\prime \\ \Phi(T,S) K_S(t) &  K_T(t)\end{bmatrix} 
\end{align*}
and
\begin{align*}  \mu&=\begin{bmatrix}\mu_1 \\ \mu_2\end{bmatrix}= \begin{bmatrix} g_S(t)+\Phi(S,t) x \\ g_T(S)+\Phi(T,S) g_S(t) +\Phi(T,t) x \end{bmatrix}= \begin{bmatrix} g_S(t) + \Phi(S,t) x \\ g_T(t)+\Phi(T,t) x\end{bmatrix}.
\end{align*}
Therefore, conditional on $X_t=x$,
\begin{equation}
\label{eq:LX_S+X_T} \begin{bmatrix}L X_S +\eta \\ X_T \end{bmatrix} = \begin{bmatrix} L & 0_{m \times d} \\ 0_{d \times d} & I_{d \times d} \end{bmatrix} \begin{bmatrix} X_S \\ X_T \end{bmatrix} +\begin{bmatrix} \eta \\ 0 \end{bmatrix} \sim 
N_{m+d}\left(\begin{bmatrix} L\mu_1  \\ \mu_2\end{bmatrix}, U(t)^{-1}  
\right), 
\end{equation}
where $U(t)$ denotes the precision matrix, defined in equation \eqref{eq:U}.
This implies that
\[
\tilde R(t,x)  = 
	-\frac{m+d}{2}\log (2\pi)-\frac12\log \left|U(t)^{-1}\right|      - \frac12 \begin{bmatrix} v_S-L\mu_1 \\ v_T -\mu_2\end{bmatrix}^\prime U(t) \begin{bmatrix} v_S-L\mu_1 \\ x_T -\mu_2\end{bmatrix}. 
\]
 It may appear that $x$ does not show up in the formula, but is appears in both $\mu_1$ and $\mu_2$.  Next, we need to take the gradient with respect to $x$. This gives 
\[ \tilde{r}(t,x)=  \begin{bmatrix} L \Phi(S,t) \\ \Phi(T,t) \end{bmatrix}^\prime U(t)\begin{bmatrix} v_S-L \mu_1\\ x_T -\mu_2\end{bmatrix}. \] Negating and differentiating once more yields the expression for $\tilde{H}$. 

\subsection{Proof of corollary \ref{corr:pull-witherror}}
\label{sec:corr:pull-witherror}

We have $\Phi(s,t)=I$ for all $s$ and $t$. This implies
	\[ U(t)^{-1} = (S-t) \begin{bmatrix} L\tilde{a} L' +(S-t)^{-1} \Sigma& L\tilde{a} \\ \tilde{a} L' &\frac{T-t}{S-t}  \tilde{a} \end{bmatrix} . \]
The Schur complement of this matrix  is 
\begin{align*}& \left((S-t)L\tilde{a} L^\T + \Sigma -  (S-t) L \tilde a  \tfrac1{T-t}\tilde{a}^{-1} \tilde a L^\T(S-t) \right)^{-1}  \\ & \quad =  \left(\frac{(S-t)(T-S)}{T-t} L \tilde a L^\T + \Sigma\right)^{-1}  =  \frac{T-t}{(S-t)(T-S)} N(t).\end{align*}
Applying the formula for the inverse of a partitioned matrix gives  
\begin{equation*} 
\begin{split}
U(t)&=\begin{bmatrix} \frac{T-t}{(S-t)(T-S)} N(t)& -\frac1{T-S} N(t)L  \\ -\frac1{T-S} L^\T N(t) & \tfrac1{T-t}\tilde{a}^{-1} -\tfrac1{T-t}\tilde{a}^{-1}\tilde{a} (S-t)L' \frac{T-t}{(S-t)(T-S)} N(t) (S-t) L \tilde{a}\tfrac1{T-t}\tilde{a}^{-1}  \end{bmatrix}\\
&=\begin{bmatrix} \frac{T-t}{(S-t)(T-S)} N(t)& -\frac1{T-S} N(t) L \\ -\frac1{T-S} L^\T N(t) & \tfrac1{T-t}\tilde{a}^{-1} -   \frac{S-t}{(T-t)(T-S)}L^\T  N(t) L   \end{bmatrix}\end{split}
 \end{equation*}
Next, we compute 
\[ W=\begin{bmatrix} W_1 & W_2 \end{bmatrix}=\begin{bmatrix} L' & I_{d\times d} \end{bmatrix} U  . \]
We have
\[ W_1= \frac{T-t}{(S-t)(T-S)} L^\T N(t)-\frac1{T-S} L^\T N(t) =\frac1{S-t} L^\T N(t) \]
and 
\begin{align*} W_2&=-\frac1{T-S} L^\T N(t) L +   \frac1{T-t} \tilde{a}^{-1} + \frac{S-t}{(T-t)(T-S)} L^\T N(t) L\\&=   \frac1{T-t} \tilde{a}^{-1}- \frac1{T-t} L^\T N(t) L.
\end{align*}
The result for $\tilde{r}$ now follows upon computing
\[ W_1 (v_S-L g_S(t)-Lx) + W_2 (x_T- g_T(t) - x). \]
The expression for $\tilde{H}$ follows from \[ \tilde{H}(t)=\begin{bmatrix} L^\T & I\end{bmatrix} U(t) \begin{bmatrix} L \\ I \end{bmatrix} =\begin{bmatrix} W_1 & W_2\end{bmatrix}\begin{bmatrix} L \\ I \end{bmatrix}  . \]
 To assess the  behaviour of the pulling term in equation \eqref{eq:noisy-pull} as $t\uparrow S$, we  write
\[	N(t)=(S-t) \left( (S-t) L\tilde{a} L^\T +\frac{T-t}{T-S} \Sigma\right)^{-1}. \]
Hence it follows that
\[ \frac{Q(t)}{S-t}= L^\T \frac{N(t)}{S-t} L \to L^\T \Sigma^{-1} L, \qquad t\uparrow S \]
and 
\[ \frac{Q(t)}{T-t} =  \frac{Q(t)}{S-t}\frac{S-t}{T-t} \to O, \qquad t\uparrow S \]
with $O$ denoting a matrix with zeroes. 

\textbf{Acknowledgement.}
M.\,S.~is supported by the European Research Council under ERC Grant Agreement 320637.

\bibliographystyle{harry}
\bibliography{lit}

\end{document}